%% file: main.tex
\documentclass{article}

\usepackage{arxiv}

\usepackage{microtype}
\usepackage{graphicx}
\usepackage{subfigure}
\usepackage{booktabs} 

\usepackage[utf8]{inputenc}
\usepackage[T1]{fontenc}
\usepackage[colorlinks=true,linkcolor=blue, citecolor=blue]{hyperref}
\usepackage{url} 
\usepackage{booktabs}
\usepackage{amsfonts} 
\usepackage[export]{adjustbox}
\usepackage{amsmath}
\usepackage{amsthm}
\usepackage{nicefrac} 
\usepackage{microtype} 
\usepackage[flushleft]{threeparttable}
\usepackage{array, multirow}
\usepackage{placeins}
\usepackage{mathtools}
\usepackage{algorithm, algorithmic}
\usepackage{enumerate}
\usepackage{bm}
\input{macros}

\begin{document}

\title{On the Size of Minimal Separators for Treedepth Decomposition}
\author{
    Zijian Xu,
    Vorapong Suppakitpaisarn\\
    The University of Tokyo\\
    \texttt{xuzijian@ms.k.u-tokyo.ac.jp}, \\
    \texttt{vorapong@is.s.u-tokyo.ac.jp} \\
}

\date{}

\maketitle

\input{contents/0-abstract}
\input{contents/1-introduction}
\input{contents/2-preliminary}
\input{contents/3-main-theorem}
\input{contents/4-special-classes}
\input{contents/5-conclusion}

\bibliography{main}
\bibliographystyle{siam}

\end{document}

%% file: macros.tex
\newtheorem{@theorem}{Theorem}[section]
\newenvironment{theorem}{\begin{@theorem}}{\end{@theorem}}
\newtheorem{lemma}{Lemma}[section]

\newtheorem{proposition}{Proposition}[section]

\newtheorem{Definition}{Definition}[section]

%% file: contents/0-abstract.tex
\begin{abstract}
\small\baselineskip=9pt
Treedepth decomposition has several practical applications and can be used to speed up many parameterized algorithms. There are several works aiming to design a scalable algorithm to compute exact treedepth decompositions. Those include works based on a set of all minimal separators. In those algorithms, although a number of minimal separators are enumerated, the minimal separators that are used for an optimal solution are empirically very small. Therefore, analyzing the upper bound on the size of minimal separators is an important problem because it has the potential to significantly reduce the computation time. A minimal separator $S$ is called an optimal top separator if $td(G) = |S| + td(G \backslash S)$, where $td(G)$ denotes the treedepth of $G$. Then, we have two theoretical results on the size of optimal top separators. (1) For any $G$, there is an optimal top separator $S$ such that $|S| \le 2tw(G)$, where $tw(G)$ is the treewidth of $G$. (2) For any $c < 2$, there exists a graph $G$ such that any optimal top separator $S$ of $G$ have $|S| > c \cdot tw(G)$, i.e., the first result gives a tight bound on the size of an optimal top separator.
\end{abstract}

%% file: contents/1-introduction.tex
\section{Introduction}

Treedepth decomposition, also known as vertex ranking number, cycle rank, or minimum height of elimination tree, is an important combinatorial optimization problem because of its applications to VLSI design \cite{leiserson1980area,sen1992graph} and numerical algorithms \cite{liu1990role}. For a particular classes of graphs, when we have an optimal treedepth decomposition of the graphs, we can have a faster algorithm for classical problems such as \textsc{WeightedMatching}, \textsc{MinimumWeightCycle}, \textsc{SteinerTree}, and \textsc{FeedbackVertexSet} \cite{sugoi-parameterized-td, wata-fully-fpt}. 

Solving the treedepth decomposition problem is NP-hard, even for chordal graphs \cite{treedepth-chordal}, but there are many works aiming to propose algorithms with small computational complexity. Those include an exact exponential algorithm whose computation time is $O^*(1.9602^n)$ in \cite{treedepth-exact-opt}, where $O^*$ notation hides the polynomial factor, an algorithm based on tree decompositions of input graphs in \cite{treedepth-parameterized-tw}, and an algorithm based on vertex cover solutions of input graphs in \cite{treedepth-vc}.

In addition to algorithms with small complexity, algorithms that can solve treedepth decomposition exactly, and are scalable in experiments are also proposed in many recent works. For example, an algorithm based on SAT solver is proposed in \cite{sat-encoding, exact-treedepth-sea}. At a competition called PACE 2020 \cite{PACE2020}, participants are asked to submit exact and scalable algorithms for the problem. The development of algorithms for treedepth decomposition has been significantly advanced there. The most scalable software could solve the problem only when the input graph has no more than $30$ nodes before, while many solvers can solve up to hundreds of nodes at the competition.

Many algorithms submitted to the PACE 2020 competition enumerate minimal separators as a subroutine. Especially, the second winning solver \cite{laakeri} and the fifth winning solver \cite{pace-xu} are based on the following theorem.

\begin{theorem}[\cite{treedepth-dp-minsep}]
Let $G$ be a graph that is not complete. Let $\Delta_G$ be the set of all minimal separators of $G$. Then, the treedepth $td(G)$ is
\begin{align}\label{eq:treedepth-rec}
    td(G) = \min\limits_{S \in \Delta_G} |S| + td(G \backslash S)
\end{align}
\end{theorem}

When the treedepth is computed by Equation \eqref{eq:treedepth-rec}, the bottleneck is often the computation of $\Delta_G$ because the number of minimal separators can be exponential with respect to the number of nodes (See Figure \ref{fig:exp-sep}). In \cite{pace-xu}, based on an observation that some of the $S$ which achieves the minimum in Equation \eqref{eq:treedepth-rec} are usually very small, a conjecture is proposed that there exists a minimal separator of size at most treewidth and it achieves the minimum in Equation \eqref{eq:treedepth-rec}. Formally, for a graph $G$ that is not complete, let $\Delta_G^* := \{S \in \Delta_G \mid td(G) = |S| + td(G \backslash S)\}$. Then, for some $S \in \Delta_G^*$, they conjecture that $|S| \le tw(G)$, where $tw(G)$ denotes the treewidth of $G$. This conjecture is important because if it is true, we can significantly reduce the number of minimal separators that we have to enumerate and can suppress the computation time.

In this paper, we answer this conjecture by the following two theorems.
\begin{theorem}\label{thm:main-1}
Let $G$ be a graph that is not complete. Then, for some $S \in \Delta_G^*$, $|S| \le 2tw(G)$.
\end{theorem}

\begin{theorem}\label{thm:main-2}
For any $c < 2$, there exists a graph $G$ such that for any $S \in \Delta_G^*$, $|S| > c\cdot tw(G)$.
\end{theorem}

By Theorem \ref{thm:main-2}, we show that the conjecture is false.

\subsection{Paper Organization}
In Section \ref{sec:preliminary}, we summarize some important concepts that are used throughout this paper, which include the definitions of the treedepth, the treewidth and minimal separators.

In Section \ref{sec:main-theorems}, we give the proofs for Theorem \ref{thm:main-1} and Theorem \ref{thm:main-2}. The proof of Theorem \ref{thm:main-2} is constructive.

In Section \ref{sec:special}, we give some examples of graph classes that have smaller upper bound on the size of $S \in \Delta_G^*$ than $2tw$. Specifically, we show that for chordal graphs, outerplanar graphs and cographs, the size of minimal separators are at most $tw$.

%% file: contents/2-preliminary.tex
\section{Preliminary}\label{sec:preliminary}
In this section, we define the treewidth, the treedepth, minimal separators, and the top separator.
\subsection{Notation}
In this paper, $G$ denotes an undirected unweighted graph. $V(G)$ or simply $V$ denote the vertex set. We use $n$ and $m$ for the number of nodes and edges, respectively. 

For a vertex set $S \subseteq V$, $G[S]$ is the subgraph induced by $S$. We use $G \backslash S$ for the graph obtained from $G$ by removing $S$, that is, $G \backslash S = G[V \backslash S]$. When $S = \{v\}$, we simply write $G \backslash v$ for short. Lastly, we write $\mathcal{C}(G)$ to denote the set of connected components of (possibly connected) graph $G$. A graph $H$ is called a minor of $G$ if $H$ can be obtained from $G$ by contracting some edges, removing some edges, and removing some isolated vertices.

\subsection{Treewidth}

Treewidth of $G$, denoted by $tw(G)$ or $tw$, is the number to show how much $G$ is close to being a tree. The number $tw(G)$ is one when $G$ is a tree and it is as large as $|V| - 1$ when $G$ is a completed graph. 

Before giving a definition of treewidth, we define tree decomposition in the following definition.

\begin{Definition}[Tree decomposition] A tree decomposition of a graph $G$ can be defined as $(\mathcal{T}, f)$ where $\mathcal{T}$ is a tree and $f$ is a function from $V(\mathcal{T})$ to $2^{V(G)}$ with the following properties:
\begin{enumerate}
\item $\bigcup\limits_{\tau \in \mathcal{T}} f(\tau) = V(G)$;
\item For each edge ${u,v}$ of $G$, there is a node $\tau \in V(\mathcal{T})$ such that ${u,v} \subseteq f(\tau)$;
\item For each $v \in V(G)$, if $\mathcal{T}_v$ is a subtree of $\mathcal{T}$ induced by the set of nodes $\{\tau \in V(\mathcal{T}): v \in f(\tau)\}$, then $\mathcal{T}_v$ is connected.
\end{enumerate}
\end{Definition}

For each $\tau \in V(\mathcal{T})$, we call the node set $f(\tau)$ as a bag of $\mathcal{T}$. We denote the maximum bag size of a tree decomposition $(\mathcal{T}, f)$ by $b(\mathcal{T}, f) := \max\limits_{\tau \in V(\mathcal{T})} |f(\tau)|$. Treewidth of $G$ is then can be defined as in the following definition:

\begin{Definition}[Treewidth] A tree decomposition $(\mathcal{T}^*, f^*)$ is an optimal tree decomposition of $G$ if, for any tree decomposition $(\mathcal{T}, f)$, $b(\mathcal{T}, f) \geq b(\mathcal{T}^*, f^*)$. Treewidth of $G$ or $tw(G)$ is the maximum bag size of $(\mathcal{T}^*, f^*)$, i.e. $tw(G) := b(\mathcal{T}^*, f^*)$. 
\end{Definition}

By the definition of tree decomposition, we have the following proposition for unconnected graph $G$.

\begin{proposition}\label{prop:tw-not-connected}
Let $G$ be a graph that is not connected. Then, the treewidth of $G$ is
\begin{align*}
    tw(G) = \max\limits_{C \in \mathcal{C}} tw(G[C]).
\end{align*}
\end{proposition}
\begin{proof}
An optimal tree decomposition of $G$ can be constructed by connecting each decomposition for the connected components to make it a tree.
\end{proof}

Next, we define the treedepth, which is the main topic of this paper. The treedepth of $G$, denoted by $td(G)$ or $td$ is the number to show how much $G$ is close to a star graph.

\begin{Definition}[Rooted forest]
A rooted graph $T$ is called a rooted forest is any connected component of $T$ is a rooted tree. The height of $T$ is defined as the maximum height among the rooted trees in $T$.
\end{Definition}

A treedepth decomposition of $G$ is defined as follows:
\begin{Definition}[Treedepth decomposition] 
Let $G$ be a connected graph. A rooted forest $T$ is called a treedepth decomposition of $G$ if
\begin{enumerate}
    \item $V(T) = V(G)$.
    \item For any $(u, v) \in E(G)$, $u$ and $v$ satisfies ancestor-descendant condition in $T$, that is, there exists a rooted tree $T'$ in $T$ and $u$ is an ancestor of $v$ or $v$ is an ancestor of $u$ in $T'$.
\end{enumerate}
\end{Definition}

By the definition of treedepth decomposition, we can define the treedepth of graph $G$ in the following definition.

\begin{Definition}[Treedepth]
The treedepth of graph $G$, denoted by $td(G)$ is the minimum height among all treedepth decompositions of $G$.
\end{Definition}

By the definition of the treedepth decomposition and the treedepth, we have the following proposition.
\begin{proposition}\label{prop:td-not-connected}
Let $G$ be a graph that is not connected. Then, the treedepth of $G$ is
\begin{align*}
    td(G) = \max\limits_{C \in \mathcal{C}(G)}td(G[C])
\end{align*}.
\end{proposition}

Unless mentioned otherwise, we assume that $G$ is a connected graph in this paper. However, by Proposition \ref{prop:tw-not-connected} and Proposition \ref{prop:td-not-connected}, it is easy to extend the discussion for $G$ that is not connected.

A treedepth decomposition of $G$ is called \emph{optimal} if its height is equal to the treedepth of the graph. 

It is known that, for any graph $G$, $td(G) \geq tw(G) + 1$ (See \cite{approximation-all, treedepth-exact-opt} for detailed explanation).

\subsection{Separators and Minimal Separators}
A node set $S \subseteq V$ is called an $a$-$b$ separator if $a, b \in V$ are not connected in $G \backslash S$. An $a$-$b$ separator $S$ is called a minimal $a$-$b$ separator if any proper subset of $S$ is not an $a$-$b$ separator. A node set $S \subseteq V$ is called a minimal separator if $S$ is a minimal $a$-$b$ separator for some $a, b \in V$. We denote the set of all minimal separators of $G$ by $\Delta_G$.

\begin{figure}[t]
    \centering
    \includegraphics[width=0.2\textwidth]{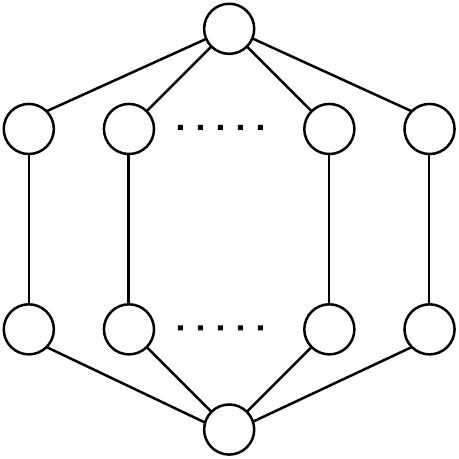}
    \caption{A graph with $tw = 2$ that has exponential number of minimal separators.}
    \label{fig:exp-sep}
\end{figure}

A graph may have an exponential number of minimal separators with repect to its node size, even when the treewidth is small (See Figure \ref{fig:exp-sep} for example). In \cite{takata-minsep} an algorithm for enumerating all separators are proposed. The running time of the algorithm is in $O(n^3 m)$ per separator. 
The algorithm is later modified in \cite{tamaki2019} to enumerate only minimal separators with bounded size. The time complexity of this modified algorithm is not given in \cite{tamaki2019}, but it is practically fast and is used in many software \cite{laakeri,tamaki2019}.

\subsection{Calculating Treedepth Using Minimal Separators} \label{naive-dp}
It is discussed in \cite{treedepth-dp-minsep} that we can compute an optimal treedepth decomposition by determining separators in a top-down way.
\begin{align}\label{eq:td-dp-ts}
    td(G) =
    \begin{cases}
    |V| & \text{if } G \text{ is complete} \\
    \min\limits_{S \in \Delta_G} \left(|S| + \max\limits_{C \in \mathcal{C}(G \backslash S)}td(G[C]) \right) & \text{otherwise}
    \end{cases}
\end{align}

\begin{figure}[t]
    \centering
    \includegraphics[width=0.6\textwidth]{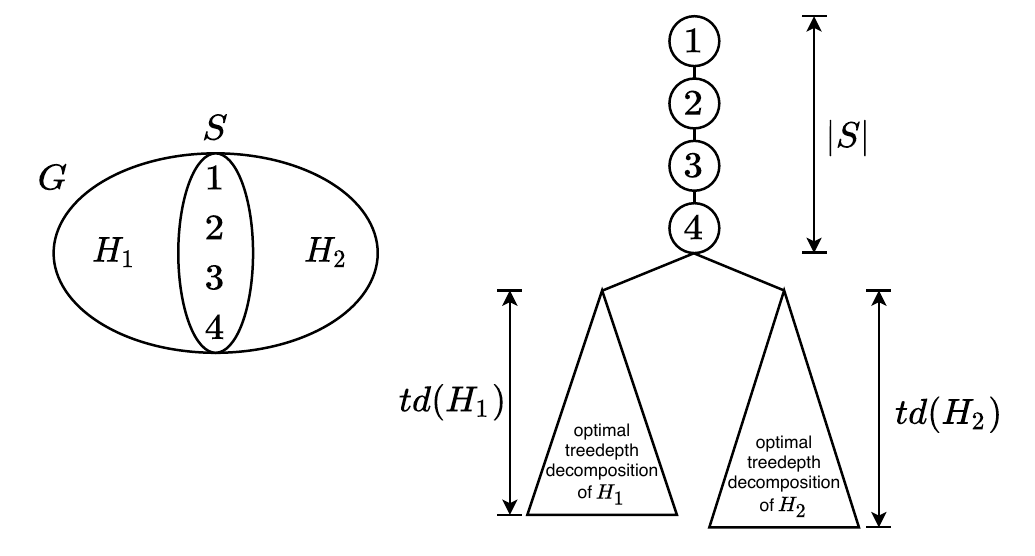}
    \caption{Optimal treedepth decomposition obtained from Equation \eqref{eq:td-dp-ts}}
    \label{fig:td-decomp-ms}
\end{figure}

To calculate an optimal treedepth decomposition from Equation \eqref{eq:td-dp-ts}, the authors begin by finding the set $\Delta_G$. Then, for each minimal separator $S \in \Delta_G$ and for each connected component $C \in \mathcal{C}(G \backslash S)$, they recursively calculate $td(G[C])$. Then, for a graph $G$ that is not complete, we obtain a separator $S^* \in \Delta_G^* := \{S \in \Delta_G \mid td(G) = |S| + \max\limits_{C \in \mathcal{C}(G \backslash S)}td(G[C])\}$. An optimal treedepth decomposition obtained from the algorithm is a tree which:
\begin{enumerate}
    \item the top of the tree is a simple path consisting of all nodes in $S^*$;
    \item the bottom end of the simple path have several branches, each of the branches is connected to the root of an optimal treedepth decomposition for $C \in \mathcal{C}(G \backslash S)$, which can be computed recursively by the same algorithm.
\end{enumerate}
We illustrate the above algorithm in Figure \ref{fig:td-decomp-ms}.

\subsection{Top Separator}
Since we are interested in the size of $S$ that appears in Equation \eqref{eq:td-dp-ts}, we define the top separator for treedepth decompositions.

\begin{Definition}
Let $T$ be a treedepth decomposition. For convenience, if $T$ is a path, we define the top separator of $T$, denoted as $ts(T)$ as $ts(T) := V(T)$. Otherwise, let $depth(T, v)$ denote the depth of $v$ in $T$ and let $p$ be the vertex that has more than one children and has smallest depth. Then, the top separator of $T$ is defined as $ts(T) := \{v \mid depth(T, v) \le depth(T, p)\}$, i.e., $ts(T)$ is the nodes that lie on the path between the root of $T$ and $p$, inclusive.
\end{Definition}

The following proposition is straightforward from Equation \eqref{eq:td-dp-ts}.
\begin{proposition}
There exists a treedepth decomposition whose top separator is a minimal separator.
\end{proposition}

If $T$ is an optimal treedepth decomposition, $ts(T)$ is called an optimal top separator.

%% file: contents/3-main-theorem.tex
\section{Main Theorems}\label{sec:main-theorems}
In this section, we give two theoretical results on the size of optimal top separators.

\subsection{Proof of Theorem \ref{thm:main-1}}
First we begin with Theorem \ref{thm:main-1}, which states for any graph $G$, there is an optimal treedepth decomposition $T$ such that $ts(T) \le 2tw(G)$.

We start with the following lemma which states the existence of a balanced separator.
\begin{lemma}[Lemma 7.19 of \cite{parameterized-algorithms}]\label{lem:subset-sep}
Let $G$ be a graph and let $U \subseteq V$. Then, there exists a separator $S$ of $G$, such that
\begin{itemize}
    \item $|S| \le tw + 1$.
    \item $G[U \backslash S]$ has more than one connected components.
    \item The size of each connected component in $G[U \backslash S]$ is at most $|U|/2$.
\end{itemize}
\end{lemma}

We are now ready to prove Theorem \ref{thm:main-1}.
\begin{proof}[Proof of Theorem \ref{thm:main-1}]
For any separator $S \in \Delta_G$, let $h(S) := |S| + \max\limits_{C \in \mathcal{C}(G \backslash S)}td(G[C])$, i.e., $h(S)$ is the small height among all treedepth decompositions that have $S$ as the top separator. Consider a separator $S$ with size larger than $2tw(G)$. 
By Lemma \ref{lem:subset-sep}, there is a separator $S'$ of $G$ with a size no larger than $tw(G) + 1$, such that each connected subgraph $G'$ of $G \backslash S'$ has $|V(G') \cap S| \le |S| / 2$.
To prove this theorem, we will show that $h(S') \le h(S)$, and, hence, there is always an optimal top separator with size no larger than $2tw(G)$.

\begin{figure}[t]
    \centering
    \includegraphics[width=0.8\textwidth]{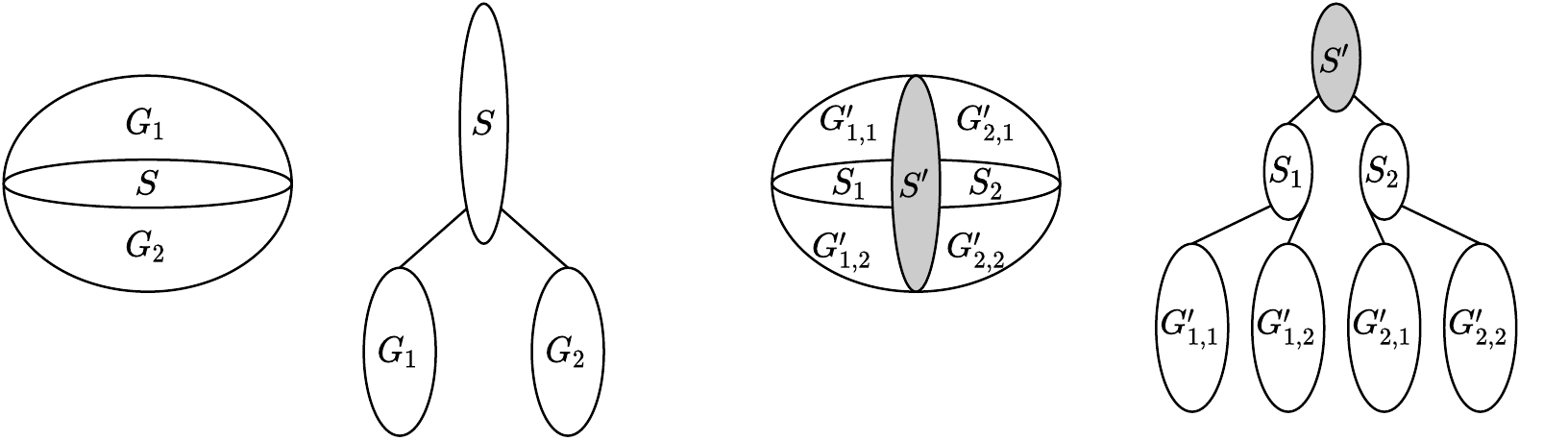}
    \caption{A rough sketch of a better decomposition using separator $S'$}
    \label{fig:proof-thm1}
\end{figure}

We consider the following treedepth decomposition $T$ (See Figure \ref{fig:proof-thm1} for inexact but helpful illustration).

\begin{itemize}
    \item The top of $T$ is a path of nodes in $S'$.

    \item The bottom node of the path has $k := |\mathcal{C}(G \backslash S')|$ branches. We denote the connected subgraphs in $\mathcal{C}(G \backslash S')$ by $G'_1, \dots G'_{k}$. A tree rooted at branch $i \in \{1, \ldots, k\}$, denoted by $T_i$, will soon be a treedepth decomposition of $G'_i$.  
    
    \item If $V(G'_i) \cap S' = \emptyset$, $T_i$ is an arbitrary optimal treedepth decomposition of $G'_i$.
    
    \item Otherwise, the top of $T_i$ is a path of nodes in $V(G'_i) \cap S$. The bottom node of the path have $|\mathcal{C}(G'_i \backslash S)|$ branches. Each of the branches is an arbitrary optimal treedepth decomposition of each connected subgraph in $\mathcal{C}(G'_i \backslash S)$.
\end{itemize}

It is straightforward to check that $T$ is a valid treedepth decomposition of $G$. Indeed, $T_i$ is actually a treedepth decomposition of $G'_i$, and, for $i \neq j$, $V(T_i) \cap V(T_j) = \emptyset$. 

The height of $T$ is
\begin{align*}
height(T) & = |S'| + \max \limits_i height(T_i)\\
& = |S'| + \max\limits_i \left[ |V(G'_i) \cap S| + \max\limits_{C \in\mathcal{C}(G'_i \backslash S)} td(G'_i[C]) \right]\\
& \le |S'| + \max\limits_i \left[ \frac{|S|}{2} + \max\limits_{C \in\mathcal{C}(G'_i \backslash S)} td(G'_i[C]) \right]\\
& \le |S'| + \frac{|S|}{2} + \max\limits_i \left[  \max\limits_{C \in\mathcal{C}(G'_i \backslash S)} td(G'_i[C]) \right]\\
& \le |S'| + \frac{|S|}{2} + \max\limits_{C \in \mathcal{C}(G \backslash S)} td(G[C]).
\end{align*}
Since $|S| \ge 2tw(G) + 1$, $|S'| \le tw(G) + 1$, we have $|S'| \le (|S| + 1)/2$. Moreover, since $|S|$ is an integer, we have
\begin{align*}
height(T) &\le (|S| + 1)/2 + |S|/2 + \max\limits_{C \in \mathcal{C}(G \backslash S)} td(G[C])\\
&= |S| + td(G \backslash S)\\
&= h(S).
\end{align*}

Therefore, we have $h(S') \le h(T) \le h(S)$ and complete the proof.
\end{proof}

\subsection{Proof of Theorem \ref{thm:main-2}}
Next, we show Theorem \ref{thm:main-2}, which states that for any $c < 2$, there exists a graph $G$ such that for any optimal treedepth decomposition $T$, $|ts(T)| > c \cdot tw(G)$, i.e., Theorem \ref{thm:main-1} gives a tight bound for the size of an optimal top separator.

The proof is by construction. We illustrate a rough sketch of the construction in Figure \ref{fig:sketch}.
\begin{figure}[t]
    \centering
    \includegraphics[height=6.5cm]{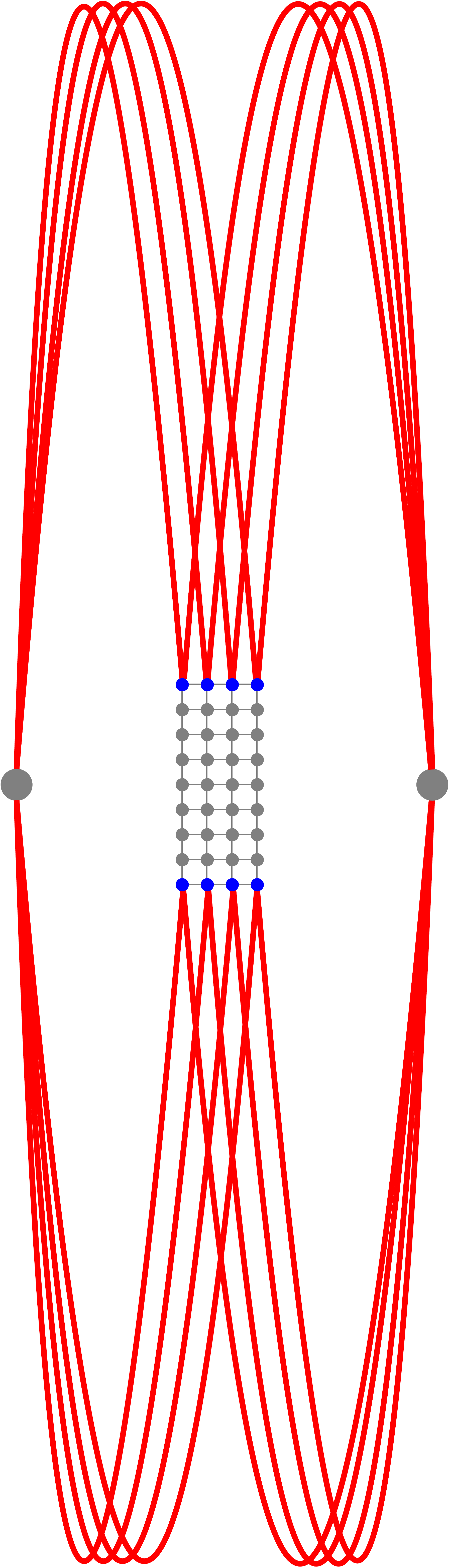}
    \caption{Sketch of the construction. Roughly, the graph is constructed by adding long paths of a certain length (red parts) to two ends of a grid. The unique optimal top separator is the blue nodes, whose size is almost twice the treewidth.}
    \label{fig:sketch}
\end{figure}

In order to analyze the treedepth and treewidth of the constructed graph, we start to define some graphs which appear as important subgraphs in our final construction.

\begin{Definition}\label{def:broom}
For positive integers $n, m, k$, we define a broom $B_{n, m, k}$ as follows (see Figure \ref{fig:broom}).
\begin{enumerate}
    \item Prepare a $P_m \times P_n$ grid such that vertices are $V_G(B_{n, m, k}) := \{(i, j) \mid 1 \le i \le n \text{ and } 1 \le j \le m\}$.
    \item For each $v \in \{(i, m) \mid 1 \le i \le n\}$, prepare a $P_{2^k - 1}$ and connect $v$ with one end of it by adding a new edge. Let $V_P(B_{n, m, k})$ be the disjoint union of the vertices in these $P_{2^k - 1}$.
\end{enumerate}

\begin{figure}[t]
    \centering
    \includegraphics[height=6.5cm]{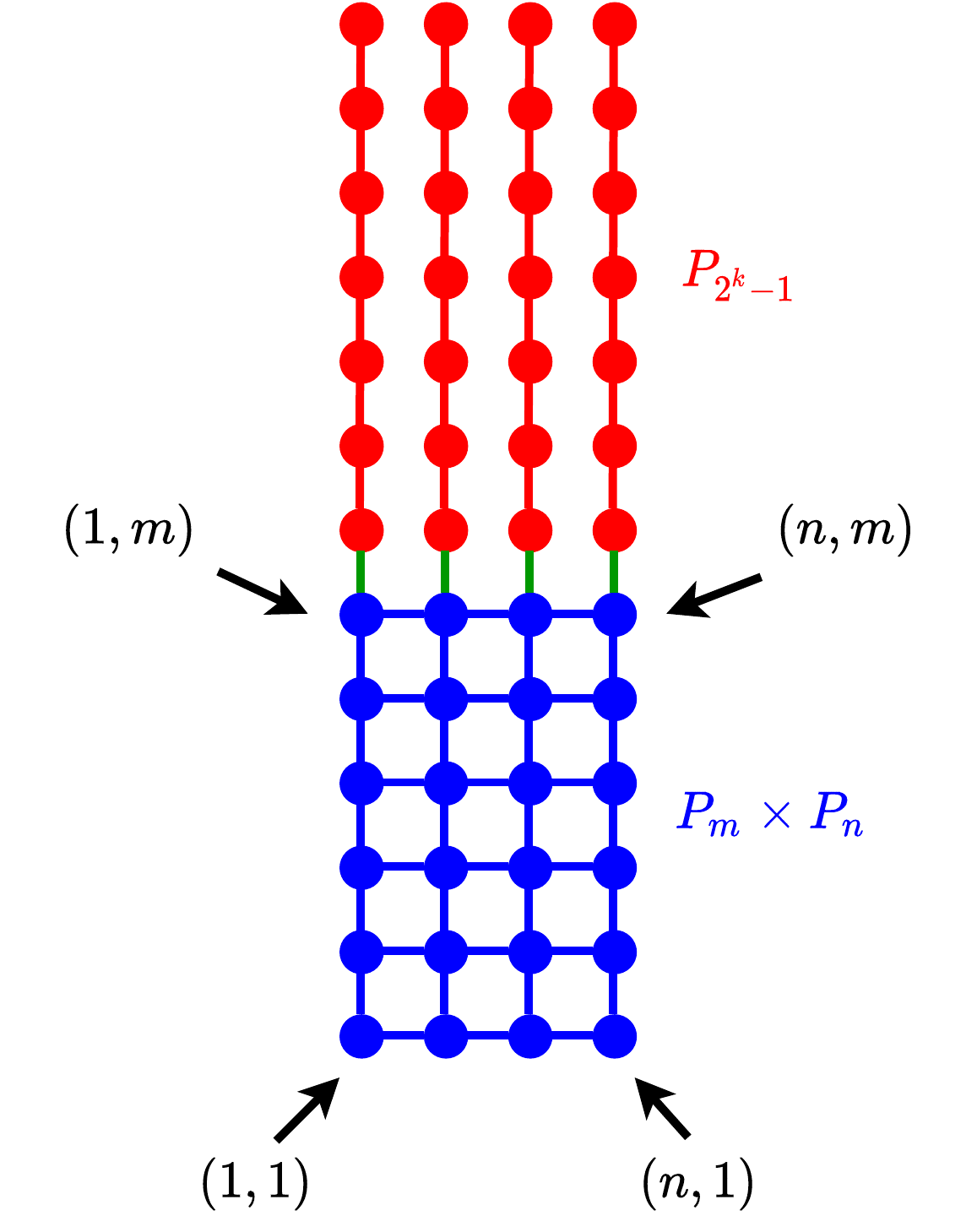}
    \caption{A broom. $V_G(B_{n, m, k})$ and $V_P(B_{n, m, k})$ are the blue vertices and the red vertices respectively.}
    \label{fig:broom}
\end{figure}
\end{Definition}

\begin{Definition}
Let $T$ be a treedepth decomposition. For $d \ge 1$, the top $d$ nodes of $T$ is defined as $top(T, d) := \{v \in V(T) \mid depth(T, v) \le d\}$, where $depth(T, v)$ is the depth of $v$ in a rooted tree $T$.
\end{Definition}

\begin{Definition}\label{def:compress}
Let $S \subseteq V(B_{n, m, k})$ such that $S \cap V_G(B_{n, m, k}) \neq \emptyset$. Define $Bottom(S) := \{(i, j) \in S \cap V_G(B_{n, m, k}) \mid j \ge j' \text{ for any } (i', j') \in S\}$. If for all $(i, j) \in Bottom(S)$, $j = m$, then define $Compress(S) := S$. Otherwise, define $Compress(S) := (S \backslash Bottom(S)) \cup \{(i, j + 1) \mid (i, j) \in Bottom(S)\}$.

\begin{figure}[t]
    \centering
    \includegraphics[height=5.5cm]{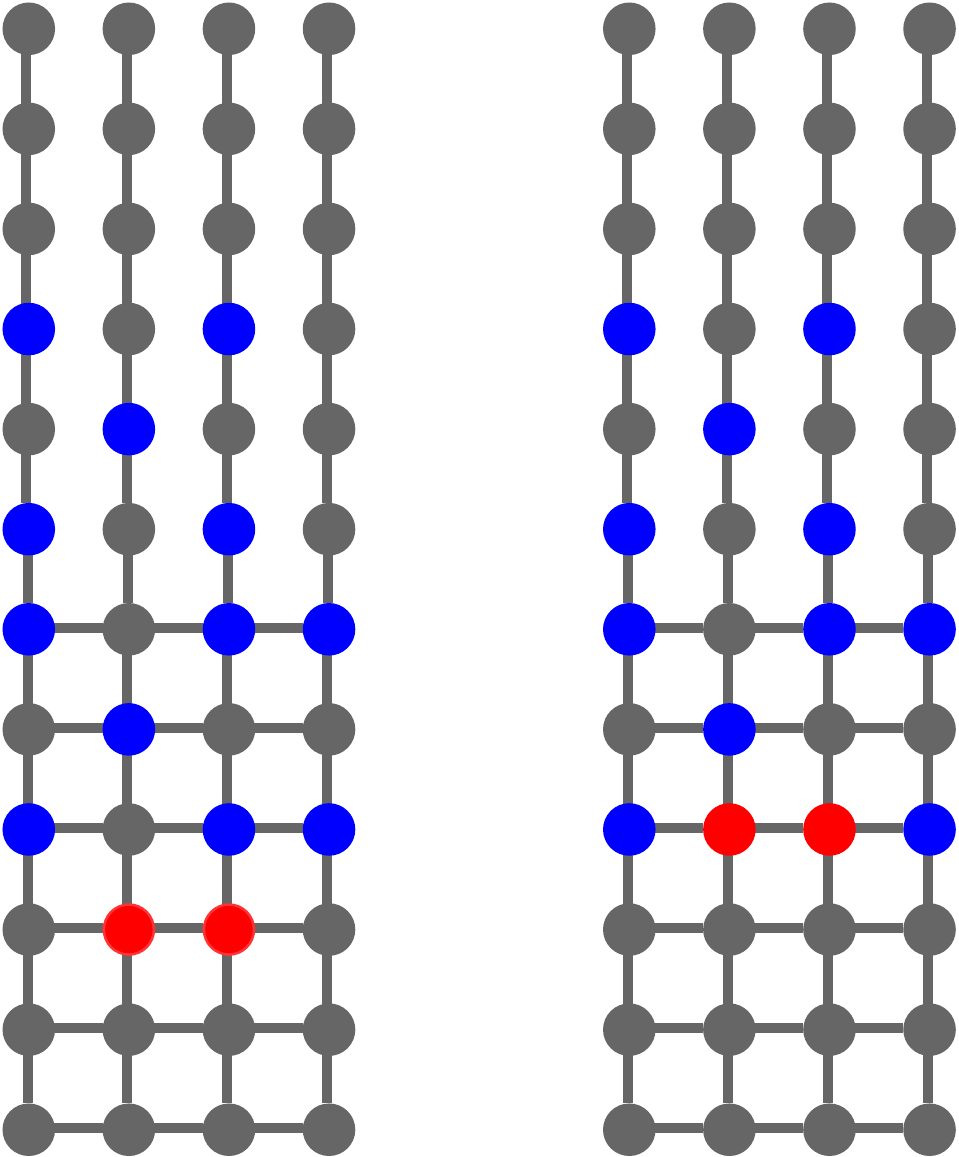}
    \caption{$S$ and $Compress(S)$ are illustrated as the colored nodes in the left and right figure, respectively. The red nodes in the left figure is $Bottom(S)$.}
    \label{fig:compress}
\end{figure}
\end{Definition}

\begin{lemma}\label{lem:broom-opt-sep}
Let $G'$ be a graph. For $k \ge 2n + td(G') + td(P_m \times P_n)$, let $G$ be a graph that is obtained by connecting $G'$ and $B_{n, m, k}$ by adding some edges between $V(G')$ and $\{(i, 1) \mid 1 \le i \le n\}$. Suppose there exists an optimal treedepth decomposition $T$ and an integer $d$ with following conditions.
\begin{itemize}
    \item There exists some $1 \le j \le m$ such that $top(T, d) \cap \{(i, j) \mid 1 \le i \le n\} = \emptyset$.
    \item For $1 \le i \le n$, let $V_i$ be the nodes that are in $\{(i, j) \mid 1 \le j \le m\}$ or in the path $P_{2^k - 1}$ that is connecting to $(i, m)$. Then, for all $1 \le i \le n$, $V_i$ contains at least one node in $top(T, d)$.
\end{itemize}
Then, $d \ge n$.
\end{lemma}
\begin{proof}
Since there exist some $1 \le j \le m$ such that $top(T, d) \cap \{(i, j) \mid 1 \le i \le n\} = \emptyset$, we can assume that $top(T, d)$ does not contain any node below this, i.e., $\{(i, j') \mid 1 \le i \le n \text{ and } j' \le j\}$ or $V(G')$ because otherwise we can take a smaller $d' < d$ which satisfies the second condition.

Let $X$ be the connected component in $G \backslash top(T, d)$ which contains $\{(i, 1) \mid 1 \le i \le n\}$ and $V(G')$. Note that decomposing $P_{2^k - 1}$ is still the bottleneck for treedepth decomposition because we have $td(G[X]) \le n + td(G') + td(P_m \times P_n)$ while $td(G[V_p(B_{n, m, k}) \backslash top(T, n)]) \ge k - n$.

Suppose $d < n$.
We are going to construct a new treedepth decomposition $T'$ in the following manner such that $h(T') \le h(T)$, i.e, $T'$ is optimal, and $Compress(top(T, d)) \subseteq top(T', d')$ for some $d' \le d$.

\begin{enumerate}
    \item Let $T' := T$.
    \item Remove $X$ from $T'$ (when $v$ is removed, the children of $v$ are connected to the parent of $v$).
\end{enumerate}
Note that at this step, $h(T') = h(T)$, otherwise it contradicts to the optimality of $T$.
\begin{enumerate}
    \setcounter{enumi}{2}
    \item If $v = (i, j) \in Bottom(top(T, d))$, replace $v$ with $(i, j + 1)$ in $T'$.
    \item If $T'$ contains more than one $v$, keep the one with smallest depth (this is unique) and remove others.
\end{enumerate}
Then, we have some $d' \le d$ such that $Compress(top(T, d)) \subseteq top(T', d')$. Finally, we modify $T'$ by
\begin{enumerate}
    \setcounter{enumi}{4}
    \item Appending an optimal treedepth decomposition of $X \cup Bottom(top(T, d))$ to an appropriate node at depth $d'$ of $T'$ (such node exists).
\end{enumerate}
Note that since $td(G[X \cup Bottom(top(T, d))]) \le n + td(G') + td(P_m \times P_n)$, this does not increase the height of $T'$.

By repeatedly applying this algorithms to construct $T'$ from $T$, we finally obtain an optimal treedepth decomposition $T^*$ such that $top(T^*, d^*)$ contains only $V_P(B_{n, m, k})$ or $\{(i, m) \mid 1 \le i \le n\}$ for some $d^* \le d$. However, in this case, we clearly have $d^* \ge n$ and therefore we have $d \ge n$, which is a contradiction.
\end{proof}

\begin{Definition}
For positive integers $n, m, k$, a double broom $D_{n, m, k}$ is defined as follows (see Figure \ref{fig:double-broom}).

\begin{enumerate}
    \item Prepare a $P_m \times P_n$ grid such that the vertices are $\{(i, j) \mid 1 \le i \le n \text{ and } 1 \le j \le m\}$.
    \item For each $v \in \{(i, j) \mid 1 \le i \le n \text{ and } j = 1, m\}$, prepare a $P_{2^k - 1}$ and connect $v$ with one end of it by adding a new edge.
\end{enumerate}

\begin{figure}[t]
    \centering
    \includegraphics[height=6.5cm]{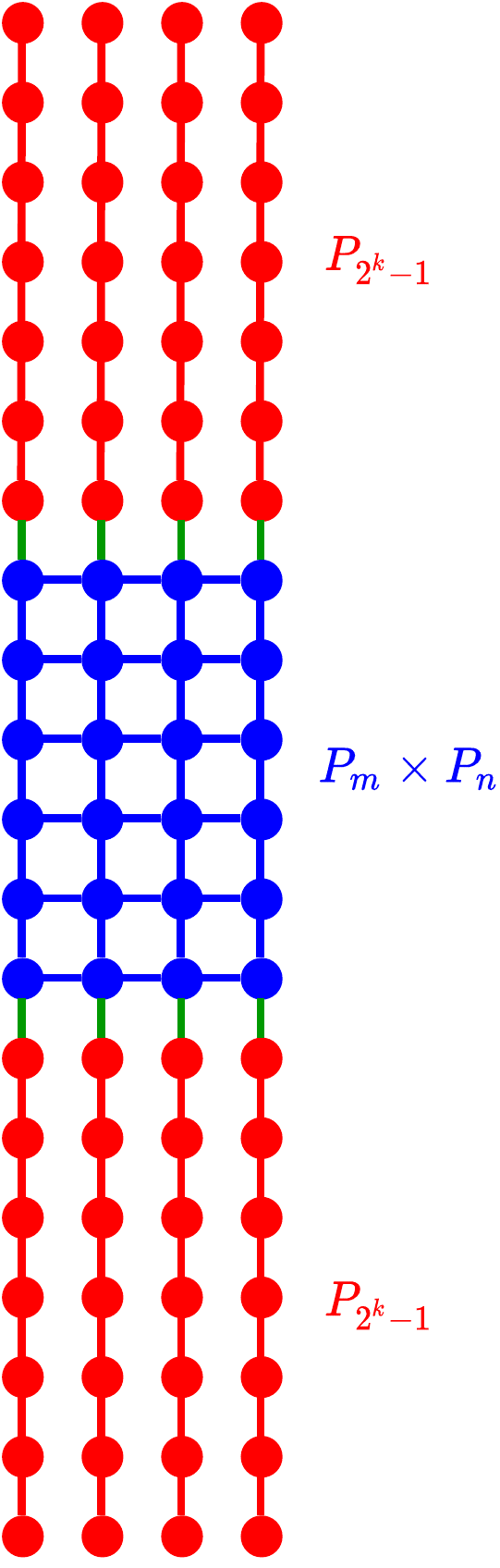}
    \caption{A double broom.}
    \label{fig:double-broom}
\end{figure}
\end{Definition}

\begin{lemma}\label{lem:double-broom-opt-sep}
Let $k \ge 2n + 2td(P_m \times P_n)$. Suppose there exists an optimal treedepth decomposition $T$ of $D_{n, m, k}$ and an integer $d$ with following conditions.
\begin{itemize}
    \item There exists some $1 \le j \le m$ such that $top(T, d) \cap \{(i, j) \mid 1 \le i \le n\} = \emptyset$.
    \item For $1 \le i \le n$, let $V_i$ be the nodes that are in $\{(i, j') \mid j' > j\}$ or in the path $P_{2^k - 1}$ that is connecting to $(i, m)$. Then, for all $1 \le i \le n$, $V_i$ contains at least one node in $top(T, d)$.
    \item For $i \le i \le n$, let $U_i$ be the nodes that are in $\{(i, j') \mid j' < j\}$ or in the path $P_{2^k - 1}$ that is connecting to $(i, 1)$. Then, for all $1 \le i \le n$, $U_i$ contains at least one node in $top(T, d)$.
\end{itemize}
Then, $d \ge 2n$. 
\end{lemma}
\begin{proof}
Let $V = \bigcup V_i$ and $U = \bigcup U_i$.
We call $top(T, 1)$ the nodes that are taken at step $1$ and for $t > 1$, We call $top(T, t) \backslash top(T, t - 1)$ the nodes that are taken at step $t$.

When both the second and the third conditions are unsatisfied, we cannot take nodes from both $V$ and $U$ in one step. Therefore, one of the smallest $d$ is obtained by a scenario that we first satisfy the second step without taking any node from $U$ and then satisfy the third condition. By Lemma \ref{lem:broom-opt-sep}, we need at least $n + n = 2n$ steps and $d \ge 2n$.
\end{proof}

\begin{lemma}\label{lem:grid-vert-sep}
Let $T$ be a treedepth decomposition of $P_{2^k - 1} \times P_n$ grid whose vertices are $\{(i, j) \mid 1 \le i \le n \text{ and } 1 \le j \le 2^k - 1\}$. Suppose that there exists a positive integer $d$ and for any $1 \le j \le 2^k - 1$, there exists some $i$ such that $(i, j) \in top(T, d)$. Then, $d \ge k$.
\end{lemma}
\begin{proof}
The proof is by mathematical induction on $k$. When $k = 1$, this is trivial. Let $s$ be the size of the top separator of $T$, i.e., $s = |ts(T)|$. We can assume that $s \le k$. Then, at least one component in $(P_{2^k - 1} \times P_n) \backslash ts(T)$ contains a grid $P_{(2^k - 1 - s) / 2} \times P_n$. In order to have $d \ge k$, we should prove that $d' \ge k - s$, where $d'$ is the $d$ defined for the smaller grid. For $s = 1$, the size of the smaller grid is $P_{2^{k - 1} - 1} \times P_n$ and by induction hypothesis, we have $d' \ge k - 1$. For $s \ge 2$, since $\frac{2^{k - s} - 1}{2} \le \frac{2^k - 1 - s}{2}$, it can be shown by applying the previous case $s$ times.
\end{proof}

\begin{lemma}\label{lem:double-broom-exact-td}
For $m \ge 2^{2n + 1} - 1$ and $k \ge 2n + 2td(P_m \times P_n)$, $td(D_{n, m, k}) = 2n + k$.
\end{lemma}
\begin{proof}
Considering a treedepth decomposition whose top separator is $\{(i, j) \mid 1 \le i \le n \text{ and } j = 1, m\}$ and then decomposing each $P_{2^k - 1}$, we obtain $td(D_{n, m, k}) \le 2n + k$.

Let $T$ be an optimal treedepth decomposition of $D_{n, m, k}$. By Lemma \ref{lem:grid-vert-sep}, there exists some $1 \le j \le m$ such that $top(T, 2n) \cap \{(i, j) \mid 1 \le i \le n\} = \emptyset$.

Let $P_i$ be the path $P_{2^k}$ which consists of $(i, m)$ and the $P_{2^k - 1}$ connecting to it. Also, let $Q_i$ be the path $P_{2^k}$ which consists of $(i, 1)$ and the $P_{2^k - 1}$ connecting to it. Then, We can assume that for all $1 \le i \le n$, $top(T, 2n)$ contains at least one node of both $P_i$ and $Q_i$, because $td(P_{2^k}) = k + 1$ and otherwise $td(D_{n, m, k}) \ge 2n + k + 1$.

By Lemma \ref{lem:double-broom-opt-sep}, we need at least $2n$ steps for this and at least for one $i$, $V_i$ or $U_i$ contains only one node in $top(T, 2n)$. Then, we have $P_{2^{k-1}}$ somewhere in $D_{n, m, k} \backslash top(T, 2n)$. Since $td(P_{2^{k - 1}}) = k$, we have $td(D_{n, m, k}) \ge 2n + k$.

Therefore, we have $td(D_{n, m, k}) = 2n + k$.
\end{proof}

\begin{Definition}
Let $G_{n, m, k, l}$ be the graph that is constructed as follows (see Figure \ref{fig:corner}).
\begin{enumerate}
    \item Prepare a $P_m \times P_n$ grid such that the vertices are $V_G(G_{n, m, k, l}) := \{(i, j) \mid 1 \le i \le n \text{ and } 1 \le j \le m\}$.
    \item Prepare $l$ new vertices $W$.
    \item For all $\{(i, j) \mid 1 \le i \le n \text{ and } j = 1, m\}$ and for all $v \in W$, prepare a $P_{2^k - 1}$. Connect one end of this path to $(i, j)$ by adding an edge, and connect the other end to $v$ by adding an edge.
\end{enumerate}

\begin{figure}[t]
    \centering
    \includegraphics[height=7cm]{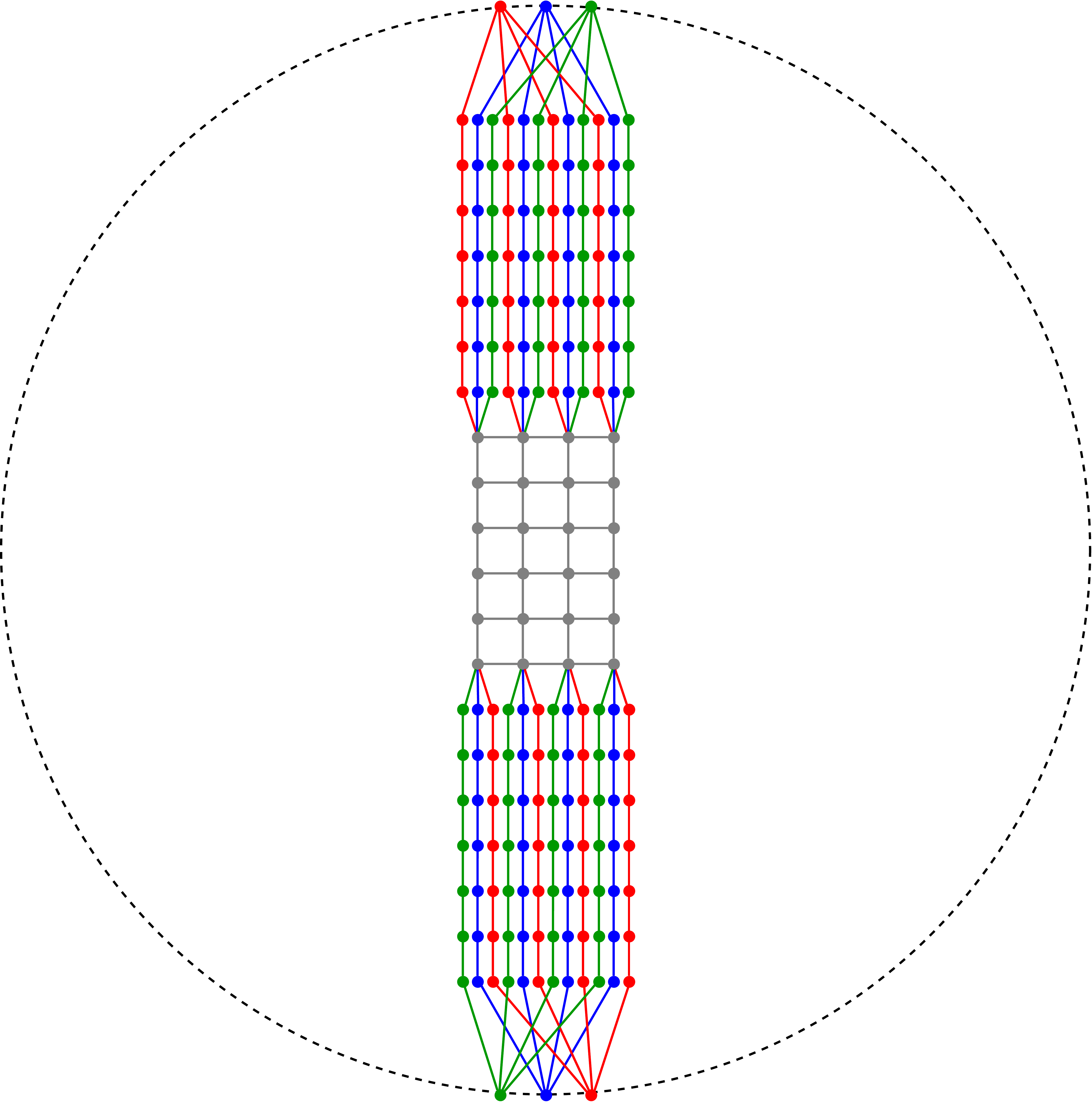}
    \caption{$G_{n, m, k, l}$. The graph is drawn on a projective plane. The vertices of the same color on the dotted cycle are identical. $W$ is the set of vertices that are on the dotted cycle.}
    \label{fig:corner}
\end{figure}
\end{Definition}

\begin{lemma}\label{lem:tw-extend-by-path}
Let $G'$ be a graph with $tw(G') \ge 1$ and let $P$ be a path. Let $G$ be the graph that is obtained by adding an edge between $u \in V(G')$ and $v \in V(P)$. Then, $tw(G) = tw(G')$.
\end{lemma}
\begin{proof}
Consider an optimal tree decomposition of $G'$ and let $X$ be one of the bags that contain $u$. Let $v_1, \ldots v_k$ the vertices that constructs $P$ such that $v_i$ and $v_{i + 1}$ are connected. We can assume that $u$ and $v_1$ are connected in $G$. Prepare new bags $X_0 = \{u, v_1\}$ and $X_i = \{v_i, v_{i + 1}\}$ for $1 \le i \le k - 1\}$. Then, extend the tree decomposition of $G'$ by connecting $X$ and $X_0$, and then connecting $X_i$ and $X_{i + 1}$ for $1 \le i \le k - 1$. Then, this decomposition has width $tw(G')$ and is an optimal tree decomposition of $G$.
\end{proof}

\begin{lemma}\label{lem:tw-of-G}
For $m \ge n$, the treewidth of $G_{n, m, k, l}$ is at most $n + l$.
\end{lemma}
\begin{proof}
By Lemma \ref{lem:tw-extend-by-path}, $tw(G_{n, m, k, l} \backslash W) = tw(P_m \times P_n) = n$. Consider an optimal tree decomposition of $G_{n, m, k, l} \backslash W$ and then extend all of the bags by adding the vertices in $W$. Then, we obtain a tree decomposition of $G_{n, m, k, l}$ of width $n + l$.
\end{proof}

\begin{lemma}\label{lem:main}
For $m \ge 2^{2n + 1} - 1$ and $k \ge 2n + 2td(P_m \times P_n)$, $\{(i, j) \mid 1 \le i \le n \text{ and } j = 1, m\}$ is the unique top separator of optimal treedepth decomposition of $G_{n, m, k, 3}$.
\end{lemma}
\begin{proof}
By considering a treedepth decomposition whose top separator is $\{(i, j) \mid 1 \le i \le n \text{ and } j = 1, m\}$ and then remove $W$, $td(G_{n, m, k, l}) \le 2n + 1 + k$.

Since $2n + k = td(D_{n, m, k}) \le td(G_{n, m, k, 1}) \le td(G_{n, m, k, 2}) \le td(G_{n, m, k, 3}) \le 2n + 1 + k$, we have $td(G_{n, m, k, 2}) = td(G_{n, m, k, 3})$.

Then, for all optimal treedepth decomposition of $G_{n, m, k, 3}$, the top separators only contain the vertices in $V_G(G_{n, m, k, 3})$.

We show the lemma by induction on $n$. When $n = 1$, we have to take two nodes from $V_G(G_{1, m, k, 3})$ to get a separator and the only optimal top separator is $\{(1, 1), (1, m)\}$.

Suppose we have shown the lemma for $n \le n'$. Let $T$ be an optimal treedepth decomposition for $G_{n' + 1, m, k, 3}$. Let $v = (i_v, j_v) \in ts(T)$ which is removed in step $1$. Then, we have $G_{n', m - 1, k, 3}$ as a minor of $G_{n' + 1, m, k, 3} \backslash v$. The optimal top separator of this graph have to contain $\{(i, j) \mid i \neq i_v \text{ and } j = 1, m\}$ because otherwise it contradicts to the induction hypothesis. To get a separator, we have to take at least two more nodes and the only choice is $\{(i_v, 1), (i_v, m)\}$ and the case is shown for $n = n' + 1$.

Therefore, for all $n$, any optimal treedepth decomposition of $G_{n, m, k, 3}$ has a unique top separator.
\end{proof}

Theorem \ref{thm:main-2} is shown as a corollary of Lemma \ref{lem:main}.
\begin{proof}[Proof of Theorem \ref{thm:main-2}]
Consider $G_{n, m, k, 3}$ for $m \ge 2^{2n + 1} - 1$ and $k \ge 2n + 2td(P_m \times P_n)$.
By Lemma \ref{lem:tw-of-G}, we have $td(G_{n, m, k, 3}) \le n + 3$, but any optimal top separator has size $2n$.
\end{proof}

%% file: contents/4-special-classes.tex
\section{Special Graph Classes}\label{sec:special}
We have shown that the upper bound on the size of an optimal top separator is $2tw$ and this is tight. In this section, we give some graph classes such that the upper bound is smaller than $2tw$. Those graph classes include chordal graphs and outerplanar graphs.

\subsection{Chordal Graphs}
A cycle is chordless if there are two nodes in the cycle that are not adjacent to each others.
A graph $G$ is chordal if there is no chordless cycle of length four or more.

To handle maximal cliques and minimal separators, we introduce clique trees.
\begin{Definition}[Clique tree]
Let $G$ be a graph, and let $\mathcal{V}$ be the set of all maximal cliques in $G$. A clique tree is a tree $\mathcal{T} = (\mathcal{V}, \mathcal{E})$ such that, for every vertex $v$ in $G$, the set of maximal cliques containing $v$ induces a connected subtree of $\mathcal{T}$.
\end{Definition}

Following lemmas states the important characteristics of chordal graphs.
\begin{lemma}[\cite{clique-tree-chordal-equiv}]\label{lem:existence-clique-tree}
A graph $G$ is chordal if and only if there exists a clique tree of $G$.
\end{lemma}

\begin{lemma}[\cite{clique-tree-sep}]\label{lem:clique-tree-sep}
Let $S$ be a minimal separator of chordal graph $G$, and let $\mathcal{T} = (\mathcal{V}, \mathcal{E})$ be a clique tree of $G$. Then, there exist two cliques $C,C' \in \mathcal{V}$ such that $\{C,C'\} \in \mathcal{E}$ and $C_i \cap C_j = S$.
\end{lemma}

\begin{lemma}[\cite{bouchitte2001treewidth}] \label{lem:clique-treewidth}
Let $G$ be a chordal graph, and let $\mathcal{V}$ be the set of all maximal cliques in $G$. Then, we have
$$tw(G) = \max_{C \in \mathcal{V}} |C| - 1.$$
\end{lemma}

The upper bound for chordal graph is stated as the following theorem.
\begin{theorem}
Let $G$ be a chordal graph that is not complete and let $T$ be any treedepth decomposition of $G$. Then, $ts(T) \le tw(G)$.
\end{theorem}
\begin{proof}
By Lemma \ref{lem:existence-clique-tree}, there exist a clique tree $\mathcal{T}$ . By Lemma \ref{lem:clique-tree-sep}, we know that, for any minimal separator $S \in \Delta_G$, there exist  $C, C' \in \mathcal{V}$ such that $C \neq C'$ and $C \cap C' = S$. We then know by Lemma \ref{lem:clique-treewidth} that 
\begin{align*}
|S| & \le \max\{|C|, |C'|\} - 1 \\
& \le \max_{C \in \mathcal{V}} |C| - 1\\
& = tw(G).
\end{align*}
\end{proof}

\subsection{Outerplanar Graphs}
A graph $G$ is outerplanar if it has an embedding on the surface of a sphere such that every edge does not cross with each other and all vertices are on the same face of the embedding.

The following lemma characterizes outerplanar graphs in terms of the graph minor.
\begin{lemma}[\cite{syslo1979characterizations}]\label{lem:outerplanar-minor-free}
A graph $G$ is outerplanar if and only if $G$ is does not contain $K_{2, 3}$ or $K_4$ as a minor, where $K_{2, 3}$ is the complete bipartite graph between $2$ nodes and $3$ nodes, and $K_4$ is the complete graph with $4$ nodes.
\end{lemma}

\begin{Definition}[Full Component]
Let $G$ be a graph and let $S$ be a separator. A connected component of $G \backslash S$ is called a full component associated with $S$ if $N(C) = S$, where $N(C)$ is the open neighbors of $C$ in graph $G$.
\end{Definition}

The following lemma characterizes minimal separators by full components.
\begin{lemma}[\cite{golumbic2004algorithmic}]\label{lem:full-comp}
Let $G$ be a graph and let $S$ be a separator. Then, $S$ is a minimal separator if and only if there exist two components $C_1$ and $C_2$ of $G \backslash S$ such that both $C_1$ and $C_2$ are full components associated with $S$.
\end{lemma}

The upper bound for outerplanar graphs is stated as follows.
\begin{theorem}\label{thm:outerplanar}
Let $G$ be an outerplanar graph that is not complete and let $T$ be any optimal treedepth decomposition of $G$. Then, $ts(T) \le tw(G) \le 2$.
\end{theorem}
\begin{proof}
Let $G$ be an outerplanar graphs. It is well known that the treewidth of outerplanar graph is at most $2$. 
When the treewidth is $1$, the graph is a tree. We know that all minimal separators of the graph tree have size $1$. Therefore, we have the theorem for when $G$ is a tree.

From now, we consider the outerplanar graphs with treewidth \textit{exactly} $2$.
We will show that our conjecture is true for outerplanar graphs, by showing that every minimal separators of an outerplanar graph have size at most $2$.

 Let $S$ be a minimal separator of $G$. Then, by Lemma \ref{lem:full-comp}, we have at least two full components $C_1$ and $C_2$ associated with $S$.
Suppose $|S| \ge 3$.
Let $v_1 \in C_1$ and $v_2 \in C_2$. Consider a graph $G'$ that is obtained from $G$ by contracting all edges in $C_1$ and $C_2$. In $G'$, $v_1$ and $v_2$ are connected to all nodes in $S$ because $N(C_1) = N(C_2) = S$ in $G$. Then, $G'$ contains a complete bipartite graph between $U = \{v_1, v_2\}$ and $V = S$ as a subgraph. Since $|S| \ge 3$, $G'$ contains $K_{2, 3}$ as a subgraph and $G$ contains $K_{2, 3}$ as a minor. That contradicts Lemma \ref{lem:outerplanar-minor-free}.
Therefore, for any minimal separator $S$, we have $|S| \le 2$.
\end{proof}

\subsubsection{Planar Graphs}
\begin{figure}[t]
    \centering
    \includegraphics[width=0.25\textwidth]{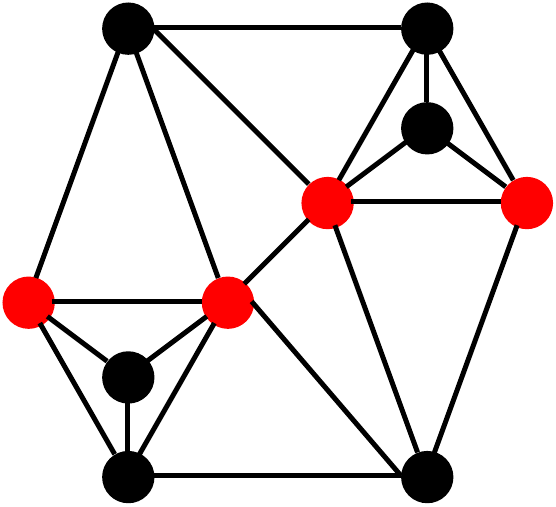}
    \caption{Planar graph with $tw = 3$ and $ts(T) = 4$ for any optimal treedepth decomposition $T$. The red nodes are one of the optimal top separators.}
    \label{fig:planar-corner}
\end{figure}
While we can upper bound the size of any optimal top separator by $tw$ for outerplanar graphs by Theorem \ref{thm:outerplanar}, this bound is not true for planar graphs. Indeed, Figure \ref{fig:planar-corner} has treewidth $3$, but any optimal treedepth decomposition of this graph has a top separator whose size is $4$. So far we do not have a tight upper bound for general planar graphs and we let this problem as future work.

\subsection{Cographs}
A graph $G$ is cograph if any of its subgraph with size $4$ is not a simple path. 

To prove our conjecture for cographs, we use the following lemma:
\begin{lemma}[\cite{treedepth-at-free}] \label{lem:cograph}
If a graph $G$ is cograph, we have
$$tw(G) = td(G) - 1.$$
\end{lemma}
We then have the following theorem, which is quite straightforward from the lemma.
\begin{theorem}
Let $G$ be a cograph that is not complete. Then, for any optimal treedepth decomposition $T$ of $G$, $ts(T) \le tw(G)$.
\end{theorem}
\begin{proof}
If $G$ is not complete, $G$ has a separator and for any optimal treedepth decomposition $T$, we have $|ts(T)| \le td(G) - 1 = tw(G)$.
\end{proof}

%% file: contents/5-conclusion.tex
\section{Conclusion}
In this paper, we proved that for any graph $G$, there is an optimal treedepth decomposition such that the size of its top separator is at most $2tw(G)$, i.e., $|ts(T)| \le 2tw(G)$. Also, we proved that this bound is tight, i.e., for any $c < 2$, there exists a graph $G$ such that for any optimal treedepth decomposition $T$ of $G$, $|ts(T)| > c \cdot tw(G)$. This answers to the previous conjecture stated in \cite{pace-xu}. We also showed a smaller upper bound on the size of an optimal top separator for some graph classes such as chordal graphs and outerplanar graphs.

%% file: main.bbl
\begin{thebibliography}{10}

\bibitem{PACE2020}
{\em {PACE 2020}}.
\newblock \url{https://pacechallenge.org/2020/}.
\newblock Accessed: 2020-08-11.

\bibitem{approximation-all}
{\sc H.~L. Bodlaender, J.~R. Gilbert, H.~Hafsteinsson, and T.~Kloks}, {\em
  Approximating treewidth, pathwidth, frontsize, and shortest elimination
  tree}, J. Algorithms, 18 (1995), pp.~238--255.

\bibitem{bouchitte2001treewidth}
{\sc V.~Bouchitt{\'e} and I.~Todinca}, {\em Treewidth and minimum fill-in:
  Grouping the minimal separators}, SIAM Journal on Computing, 31 (2001),
  pp.~212--232.

\bibitem{parameterized-algorithms}
{\sc M.~Cygan, F.~V. Fomin, {\L}.~Kowalik, D.~Lokshtanov, D.~Marx,
  M.~Pilipczuk, M.~Pilipczuk, and S.~Saurabh}, {\em Parameterized algorithms},
  vol.~4, Springer, 2015.

\bibitem{treedepth-dp-minsep}
{\sc J.~S. Deogun, T.~Kloks, D.~Kratsch, and H.~M{\"u}ller}, {\em On the vertex
  ranking problem for trapezoid, circular-arc and other graphs}, Discrete
  Applied Mathematics, 98 (1999), pp.~39--63.

\bibitem{treedepth-chordal}
{\sc D.~Dereniowski and A.~Nadolski}, {\em Vertex rankings of chordal graphs
  and weighted trees}, Information Processing Letters, 98 (2006), pp.~96--100.

\bibitem{sugoi-parameterized-td}
{\sc H.~Falko and K.~Stefan}, {\em Solving connectivity problems parameterized
  by treedepth in single-exponential time and polynomial space}, arXiv preprint
  arXiv:2001.05364,  (2020).

\bibitem{treedepth-exact-opt}
{\sc F.~V. Fomin, A.~C. Giannopoulou, and M.~Pilipczuk}, {\em Computing
  tree-depth faster than $2^n$}, Algorithmica, 73 (2015), pp.~202--216.

\bibitem{sat-encoding}
{\sc R.~Ganian, N.~Lodha, S.~Ordyniak, and S.~Szeider}, {\em {SAT-encodings}
  for treecut width and treedepth}, in ALENEX 2019, 2019, pp.~117--129.

\bibitem{clique-tree-chordal-equiv}
{\sc F.~Gavril}, {\em The intersection graphs of subtrees in trees are exactly
  the chordal graphs}, Journal of Combinatorial Theory, Series B, 16 (1974),
  pp.~47--56.

\bibitem{golumbic2004algorithmic}
{\sc M.~C. Golumbic}, {\em Algorithmic graph theory and perfect graphs},
  Elsevier, 2004.

\bibitem{clique-tree-sep}
{\sc C.-W. Ho and R.~C.~T. Lee}, {\em Counting clique trees and computing
  perfect elimination schemes in parallel}, Information Processing Letters, 31
  (1989), pp.~61--68.

\bibitem{wata-fully-fpt}
{\sc Y.~Iwata, T.~Ogasawara, and N.~Ohsaka}, {\em On the power of tree-depth
  for fully polynomial {FPT} algorithms}, arXiv preprint arXiv:1710.04376,
  (2017).

\bibitem{treedepth-at-free}
{\sc T.~Kloks, H.~M{\"u}ller, and C.~Wong}, {\em Vertex ranking of asteroidal
  triple-free graphs}, Information Processing Letters, 68 (1998), pp.~201--206.

\bibitem{treedepth-vc}
{\sc Y.~Kobayashi and H.~Tamaki}, {\em Treedepth parameterized by vertex cover
  number}, in IPEC 2016, 2017.

\bibitem{laakeri}
{\sc T.~Korhonen}, {\em {PACE Solver Description: SMS}}, in 15th International
  Symposium on Parameterized and Exact Computation (IPEC 2020), vol.~180,
  Schloss Dagstuhl--Leibniz-Zentrum f{\"u}r Informatik, 2020, pp.~30:1--30:4.

\bibitem{leiserson1980area}
{\sc C.~E. Leiserson}, {\em Area-efficient graph layouts}, in FOCS 1980, 1980,
  pp.~270--281.

\bibitem{liu1990role}
{\sc J.~W. Liu}, {\em The role of elimination trees in sparse factorization},
  SIAM Journal on Matrix Analysis and Applications, 11 (1990), pp.~134--172.

\bibitem{treedepth-parameterized-tw}
{\sc F.~Reidl, P.~Rossmanith, F.~S. Villaamil, and S.~Sikdar}, {\em A faster
  parameterized algorithm for treedepth}, in ICALP 2014, 2014, pp.~931--942.

\bibitem{sen1992graph}
{\sc A.~Sen, H.~Deng, and S.~Guha}, {\em On a graph partition problem with
  application to {VLSI} layout}, Information Processing Letters, 43 (1992),
  pp.~87--94.

\bibitem{syslo1979characterizations}
{\sc M.~M. Sys{\l}o}, {\em Characterizations of outerplanar graphs}, Discrete
  Mathematics, 26 (1979), pp.~47--53.

\bibitem{takata-minsep}
{\sc K.~Takata}, {\em Space-optimal, backtracking algorithms to list the
  minimal vertex separators of a graph}, Discrete Applied Mathematics, 158
  (2010), pp.~1660--1667.

\bibitem{tamaki2019}
{\sc H.~Tamaki}, {\em Computing treewidth via exact and heuristic lists of
  minimal separators}, in SEA 2019, 2019, pp.~219--236.

\bibitem{exact-treedepth-sea}
{\sc J.~Trimble}, {\em An algorithm for the exact treedepth problem}, in SEA
  2020, 2020, pp.~19:1--19:14.

\bibitem{pace-xu}
{\sc Z.~Xu, D.~Mao, and V.~Suppakitpaisarn}, {\em {PACE Solver Description:
  Computing Exact Treedepth via Minimal Separators}}, in 15th International
  Symposium on Parameterized and Exact Computation (IPEC 2020), vol.~180,
  Schloss Dagstuhl--Leibniz-Zentrum f{\"u}r Informatik, 2020, pp.~31:1--31:4.

\end{thebibliography}
